\documentclass[conference]{IEEEtran}
\usepackage{amsmath, amssymb}
\usepackage{graphicx, color}
\usepackage[latin1]{inputenc}
\usepackage{subcaption}

\newtheorem{theorem}{Theorem}

\newtheorem{proof}{Proof}

\def\sinr {\mbox{\scriptsize\sf SINR}}
\newcommand{\PP}{\mathbb{P}}
\newcommand{\E}{\mathbb{E}}
\newcommand{\La}{\mathcal{L}}
\newcommand{\dd}{{\rm d}}
\newcommand{\R}{\mathbb{R}}
\DeclareMathOperator{\arccot}{arccot}

\IEEEoverridecommandlockouts
\allowdisplaybreaks

\begin{document}

\title{Outage Analysis of Full-Duplex Architectures in Cellular Networks}

\author{
\IEEEauthorblockN{Constantinos Psomas and Ioannis Krikidis}
\IEEEauthorblockA{Department of Electrical and Computer Engineering, University of Cyprus, Cyprus}
\IEEEauthorblockA{e-mail: \{psomas, krikidis\}@ucy.ac.cy}
\thanks{This work was supported by the Research Promotion Foundation, Cyprus under the project KOYLTOYRA/BP-NE/0613/04 ``Full-Duplex Radio: Modeling, Analysis and Design (FD-RD)''.}\vspace{-7mm}}

\maketitle

\begin{abstract}
The implementation of full-duplex (FD) radio in wireless communications is a potential approach for achieving higher spectral efficiency. A possible application is its employment in the next generation of cellular networks. However, the performance of large-scale FD multiuser networks is an area mostly unexplored. Most of the related work focuses on the performance analysis of small-scale networks or on loop interference cancellation schemes. In this paper, we derive the outage probability performance of large-scale FD cellular networks in the context of two architectures: two-node and three-node. We show how the performance is affected with respect to the model's parameters and provide a comparison between the two architectures.
\end{abstract}


\section{Introduction}
The ever expanding world of wireless communications has motivated the need for new techniques to improve the utilization of the radio spectrum and to increase the spectral efficiency. This has been partly achieved by the use of orthogonal channel access schemes in conventional half-duplex (HD) systems. In this case, an HD wireless node transmits and receives information using orthogonal channels (e.g. frequency, time) for each operation. However, these orthogonality schemes lead to inefficient use of the system's bandwidth resources which in turn has prompted research for overcoming these limitations \cite{KRI1}. A potential solution to the HD constraints is full-duplex (FD) radio, as it allows a wireless node to simultaneously transmit and receive information at the same time and frequency. The main drawback of FD is the loop interference (LI) formed between the output and the input antennas which can be catastrophic to the system's efficiency and has been the primary reason why FD has been perceived as impractical so far. In spite of that, recent improvements in antenna technology and signal processing have helped mitigate this interference and, as a result, made FD feasible \cite{TR1}-\cite{AS3}. Indeed, FD has gained popularity recently and the literature list regarding this area has expanded significantly (see \cite{JSAC} and references therein).

In the context of wireless communication systems, FD has been studied mostly for simple topologies with a single user, and work for multiuser systems such as cellular and WiFi networks has been limited. In such networks, the simultaneous uplink and downlink operation at the same time and frequency creates multiuser interference which has a negative effect on the system's performance. The work in \cite{AS4} studies a three-node network with an FD base station and two HD mobile nodes and shows how a side-channel information can help reduce the effects of the interference from the uplink to the downlink node. In \cite{AS5}, an FD capable network with a multiple input multiple output (MIMO) base station serving multiple users is studied, where information theoretic interference management techniques achieve rate gains over an HD network. The network configuration of \cite{AS5} is also investigated in \cite{YIN} and it is shown that the network can also achieve rate gains by exploiting the degrees of freedom in a MIMO base station.

The aforementioned studies involve `static' small-scale scenarios in the sense that the distances between the nodes are not taken into account and they involve only one base station. To the authors' knowledge, `dynamic' large-scale FD networks have only been studied in \cite{XIE} and \cite{GOY}. The work in \cite{XIE} investigates a large-scale ad-hoc FD network and concludes that the large-scale factor in this case has a negative impact on the potential gains of the FD. On the other hand, benefits of FD are demonstrated in \cite{GOY}, where a large-scale FD cellular network based on the three-node architecture is studied. Tractable analytical expressions for the average per channel rate of both uplink and downlink are obtained using stochastic geometry, which show that the FD increases the aggregate throughput compared to the HD counterpart. Even though the results of \cite{GOY} are promising for the FD prospects, further investigation should be undertaken in terms of the outage probability of the system. Moreover, the two-node architecture, where both the base station and the user are FD capable, should also be taken into account. In this paper, we provide an analysis using stochastic geometry on the performance of the downlink in the two-node and the three-node FD architectures \cite{AS2}. The interference in each scenario acts differently on the downlink and so a comparison of the performances is provided together with the conventional HD one \cite{JA} to show the points at which each scenario overtakes the others. The rest of the paper is organized as follows: the next section sets forth the system model and its main assumptions. Section \ref{sec:outage} presents the analysis for the outage probability and Section \ref{sec:validation} presents a numerical validation and evaluation of the model. Finally Section \ref{sec:conclusion} provides some conclusive remarks.

\underline{Notation}: $\R^d$ denotes the $d$-dimensional Euclidean space, $b(x,r)$ denotes a two dimensional disk of radius $r$ centered at $x$, $\|x\|$ denotes the Euclidean norm of $x \in \R^d$, $N(A)$ represents the number of points in the area $A$, $\PP(X)$ denotes the probability of the event $X$ and $\E(X)$ represents the expected value of $X$.

\begin{figure}[t]
\begin{subfigure}{0.45\linewidth}
  \includegraphics[width=\linewidth]{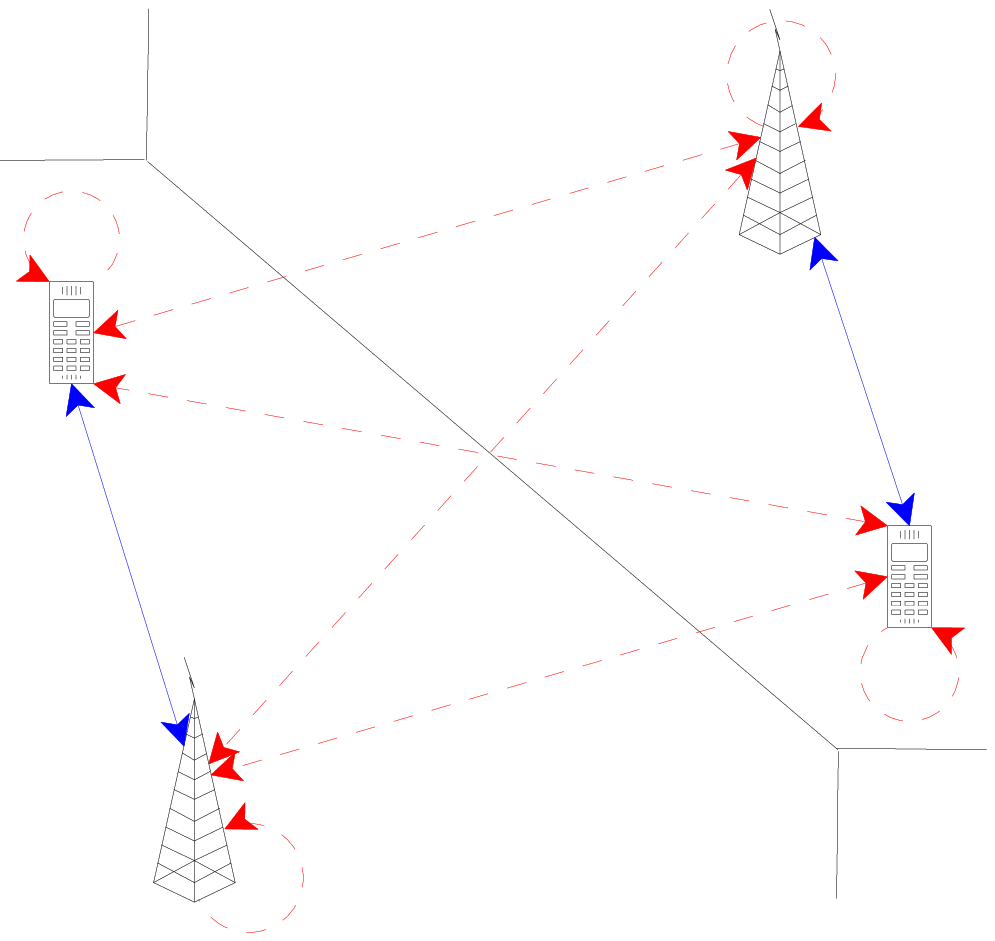}
  \caption{Two-node}
  \label{fig:sce1}
\end{subfigure}\hfill
\begin{subfigure}{0.45\linewidth}
  \includegraphics[width=\linewidth]{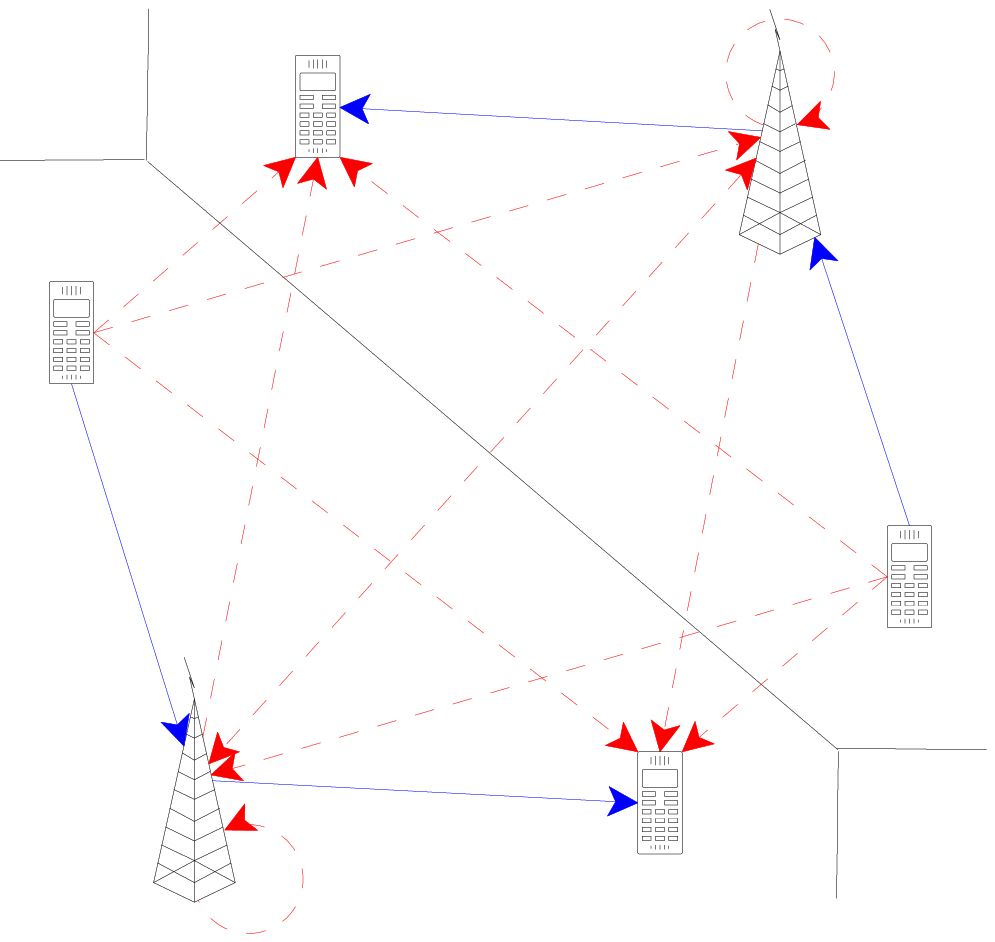}
  \caption{Three-node}
  \label{fig:sce2}
\end{subfigure}
\caption{Full-duplex architectures.}\label{fig:scenarios}
\vspace{-3mm}
\end{figure}

\section{System Model}\label{sec:model}
We consider an FD-capable cellular network and focus on the downlink performance. We take into account two different scenarios where FD can be employed: a two-node and a three-node architecture \cite{AS2}. Both scenarios are illustrated in Fig. \ref{fig:scenarios} where the active links are depicted with solid blue lines and the interference signals with dashed red lines. In the two-node FD scenario, both the base station (BS) and the user are FD capable and at any generic time a BS serves just one user for both uplink and downlink (Fig. \ref{fig:sce1}). In the three-node FD scenario, only the BS is FD capable and at any generic time a BS serves one HD uplink user and one HD downlink user (Fig. \ref{fig:sce2}). As the main focus of this work is the performance of the downlink user, we assume that in both scenarios the BS operates FD in a similar way. Both scenarios are modeled using spatial Poisson point processes (PPP) \cite{HAEN1}. For the sake of simplicity, the same notation is used for both models. Let the locations of the BSs to be distributed by a homogeneous PPP $\Phi_b = \{ x_i: i = 1,2,\dots\}$ of density $\lambda$ in the Euclidean plane $\R^2$, where $x_i \in \R^2$ denotes the location of the $i^{\rm th}$ BS. Similarly, let $\Phi_u$ be a homogeneous PPP of the same density $\lambda$ but independent of $\Phi_b$ to represent the locations of the users. We assume that all the BSs transmit with the same power $P_b$ and all the users with the same power $P_u$; both are equipped with a single transmit and a single receive antenna.

We assume that the channels are subject to both small-scale fading and large-scale path loss. Specifically, the fading between two nodes is Rayleigh distributed and so the power of the channel fading is an exponential random variable with mean $1/\mu$. The channel fadings are considered to be independent between them. The standard path loss model $\ell(x,y) = \|x-y\|^{-\alpha}$ is used which assumes that the received power decays with the distance between the transmitter $x$ and the receiver $y$, where $\alpha > 2$ denotes the path loss exponent. Throughout this paper, we will denote by $\alpha_1$ and $\alpha_2$ the path loss exponents for the channel between a BS and a user and for the channel between a pair of users respectively. In both scenarios, the interference at the downlink user is the sum of the received signals from the BSs of $\Phi_b$ and the uplink users of $\Phi_u$, excluding the received signal from the BS in the same cell. Note that in the two-node scenario the uplink and downlink operation is performed by the same user and so intra-cell interference does not exist. Nevertheless, in this case the user experiences LI. We assume that imperfect cancellation mechanisms of the LI are used \cite{TR3}, \cite{AS3} and the channel gain $h_l$ from the residual interference after cancellation can be characterized by $\E[\|h_l\|^2]=\sigma_l^2$ as each implementation of the cancellation mechanism can be characterized by a specific residual power \cite{KRI2}. Moreover, all wireless links exhibit additive white Gaussian noise (AWGN) with variance $\sigma_n^2$. A downlink user selects to connect to the BS transmitting the strongest signal power. Since all the BSs transmit with the same power, the user connects to the nearest BS in the plane. Assuming the user is located at the origin $o$ and at a distance $r$ to the nearest BS, the cumulative distribution function (cdf) of $r$ can be derived from the null probability of a 2D Poisson process \cite{HAEN2}, $\PP[r \leq R] = 1 - \PP[N(b(o,R))=0] = 1 - e^{-\lambda \pi R^2}.$
Therefore, the probability density function (pdf) of $r$ is,
\begin{equation}\label{eq:dist}
f_r(r) = 2\pi \lambda re^{-\lambda \pi r^2}, ~r \geq 0.
\end{equation}

\section{Outage Probability}\label{sec:outage}
In this section, we derive analytically the outage probability of a downlink cellular network for both scenarios outlined in Section \ref{sec:model}. The outage probability describes the probability that the instantaneous achievable rate of the channel is less than a fixed target rate $R$, i.e. $\PP[\log(1+\sinr) < R]$. Without loss of generality and following Slivnyak's Theorem \cite{STO}, we execute the analysis for a typical user $u_o$ located at the origin but the results hold for all downlink users in the network.

Assuming $u_o$ is at a random distance $r$ from the nearest BS, denoted by $b_o$, then the SINR of $u_o$ is,
\begin{equation}\label{eq:sinr}
\sinr = \frac{P_b h r^{-\alpha_1}}{\sigma_n^2 + I_l + I_b + I_u},
\end{equation}
where $I_l$ is the residual interference at $u_o$ after LI cancellation and is defined as $I_l = P_u h_l$, where $h_l$ is the LI channel gain at $u_o$; $I_b$ and $I_u$ is the interference received at $u_o$ from all the BSs (apart from $b_o$), and all the uplink users respectively. Specifically,
\begin{equation}\label{eq:inter}
I_b = P_b \sum_{i\in\Phi_b \setminus b_o} g_i d_i^{-\alpha_1}, \;\; 
I_u = P_u \sum_{j\in\Phi_u} k_j D_j^{-\alpha_2},
\end{equation}
where $h$, $g_i$, $k_j$ are the channel fadings between $u_o$ and $b_o$, $u_o$ and the $i^{\rm th}$ BS and $u_o$ and the $j^{\rm th}$ uplink user respectively; similarly, $d_i$ and $D_j$ are the distances between $u_o$ and the $i^{\rm th}$ BS and $u_o$ and the $j^{\rm th}$ uplink user respectively.

\begin{theorem}\label{thm:two}
The outage probability of a downlink user in the two-node FD scenario is $\Pi_2(R, \lambda, \alpha_1, \alpha_2) =$
\begin{equation}
1 - 2 \pi \lambda \int_0^\infty 
\frac{r e^{-\lambda \pi r^2-s\sigma_n^2}}{1+\frac{P_u}{P_b}\sigma_l^2 Tr^{\alpha_1}}
\La_{I_b}\left(s\right)
\La_{I_u}\left(s\right) \dd r,\label{eq:Thm1}
\end{equation}
where $s = \frac{\mu Tr^{\alpha_1}}{P_b}$, $T=2^R-1$,
\begin{equation}
\La_{I_b}\left(s\right) = \exp\left(-2\pi\lambda \int_r^\infty \left(\frac{T}{T + (\frac{x}{r})^{\alpha_1}} \right) x \dd x \right),
\end{equation}
and
\begin{align}\nonumber
&\La_{I_u}\left(s\right)
=\\& 2\pi \lambda\int_0^\infty \rho e^{-\lambda \pi \rho^2}
\exp\left(-2\pi\lambda \int_\rho^\infty \left(\frac{\frac{P_u}{P_b}T}{\frac{P_u}{P_b}T + \frac{y^{\alpha_2}}{r^{\alpha_1}}} \right) y \dd y \right)\dd \rho.
\end{align}
\end{theorem}

\begin{proof}
Starting from the definition of the outage probability and conditioning on the nearest BS being at a distance $r$ we have,
\begin{align*}
&\Pi_2(R, \lambda, \alpha_1, \alpha_2) =
\E_r \left[ \PP[\log(1+\sinr) < R\ |\ r] \right]\\
&=\int_0^\infty \PP[\log(1+\sinr) < R\ |\ r]f_r(r)\dd r\\
&=1-2\pi \lambda\int_0^\infty \PP[\sinr \geq 2^R-1\ |\ r] \;re^{-\lambda \pi r^2}\dd r.
\end{align*}
Letting $T=2^R-1$, $\PP[\sinr \geq T\ |\ r]$ is the coverage probability conditioned on the distance $r$ and is given by,
\begin{align*}
\PP[\sinr \geq T\ |\ r] &=\PP\left[h \geq \frac{Tr^{\alpha_1}}{P_b}(\sigma_n^2 + I_l + I_b + I_u)\ \Big|\ r\right]\\
&\stackrel{(a)}{=} \E\left[e^{-\frac{\mu Tr^{\alpha_1}}{P_b}(\sigma_N^2 + I_l + I_b + I_u)}\ \Big|\ r\right]\\
&= e^{-s\sigma_n^2}\E_{I_l}\left[e^{-sI_l}\right]
\E_{I_b}\left[e^{-sI_b}\right]\E_{I_u}\left[e^{-sI_u}\right]\\
&\stackrel{(b)}{=} \frac{e^{-s\sigma_n^2}}{1+\sigma_l^2 Tr^{\alpha_1}\frac{P_u}{P_b}} \La_{I_b}(s) \La_{I_u}(s),
\end{align*}
where $s = \frac{\mu Tr^{\alpha_1}}{P_b}$; $(a)$ follows from the fact that $h \sim \exp(\mu)$; $(b)$ follows from the moment generating function (MGF) of an exponential variable and since $h_l \sim \exp(1/\sigma^2_l)$; $\La_{I_b}(s)$ and $\La_{I_u}(s)$ are the Laplace transforms of the random variables $I_b$ and $I_u$ respectively, evaluated at $s$. As there is no intra-cell interference, $I_u$ needs to be evaluated conditioned on the distance $\rho$ from $u_o$ to the closest uplink user in the neighboring cells. Since the densities of $\Phi_b$ and $\Phi_u$ are equal, we can assume that there is on average one user per cell. Therefore, $\rho$ is distributed according to \eqref{eq:dist} and the Laplace transform of $I_u$ is given by,
\begin{equation}\label{eq:laplace_I_u}
\La_{I_u}(s) = \E_{I_u}[e^{-sI_u}\ |\ \rho] = \int_0^\infty \E_{I_u}[e^{-sI_u}]f_{\rho}(\rho)\dd \rho.
\end{equation}
The expected value is then evaluated as follows,
\begin{align}
\nonumber& \E_{I_u}[e^{-sI_u}] = \E_{\Phi_u, k_j}\left[\exp(-sP_u\sum_{j\in \Phi_u} k_j D_j^{-\alpha_2})\right]\\
\nonumber&= \E_{\Phi_u, k_j}\left[\prod_{j\in \Phi_u} \exp(-sP_u k_j D_j^{-\alpha_2})\right]\\
\nonumber&\stackrel{(a)}{=}\E_{\Phi_u}\left[\prod_{j\in \Phi_u} \E_{k} [\exp(-sP_u k D_j^{-\alpha_2})] \right]\\
\nonumber&\stackrel{(b)}{=} \exp\left(-2\pi\lambda \int_\rho^\infty \left(1-\E_{k}[\exp(-sP_u k y^{-\alpha_2})] \right) y \dd y \right)\\
&\stackrel{(c)}{=} \exp\left(-2\pi\lambda \int_\rho^\infty \left(1 - \frac{\mu}{\mu + s P_u y^{-\alpha_2}} \right) y \dd y \right),\label{eq:exp_I_u}
\end{align}
where $(a)$ follows from the fact that $k_j$ are independent and identically distributed and also independent from the point process $\Phi_u$; $(b)$ follows from the probability generating functional (PGFL) of a PPP \cite{STO} and the limits are from $\rho$ to $\infty$ since the closest interfering uplink user is at least at a distance $\rho$; $(c)$ follows from the MGF of an exponential random variable and since $k \sim \exp(\mu)$.

\noindent Replacing $\E_{I_u}[e^{-sI_u}]$ with \eqref{eq:exp_I_u} and $s$ with $\frac{\mu Tr^{\alpha_1}}{P_b}$ in \eqref{eq:laplace_I_u} gives,\vspace{0.2cm}
\begin{align*}
&\La_{I_u}\left(\frac{\mu Tr^{\alpha_1}}{P_b}\right) = 
\int_0^\infty \E_{I_u}\left[\exp\left(-\frac{\mu Tr^{\alpha_1}I_u}{P_b}\right)\right]f_{\rho}(\rho) \dd \rho =\\[0.2cm]
&\int_0^\infty \exp\left(-2\pi\lambda \int_\rho^\infty \left(1 - \frac{1}{1 + \frac{P_u}{P_b}\frac{r^{\alpha_1}}{y^{\alpha_2}}T} \right) y \dd y \right) f_{\rho}(\rho) \dd \rho =\\[0.2cm]
&2\pi \lambda\int_0^\infty \rho e^{-\lambda \pi \rho^2}
\exp\left(-2\pi\lambda \int_\rho^\infty \left(\frac{\frac{P_u}{P_b}T}{\frac{P_u}{P_b}T + \frac{y^{\alpha_2}}{r^{\alpha_1}}} \right) y \dd y \right)\dd \rho.
\end{align*}
Similarly as above,\vspace{0.2cm}
\[{\cal L}_{I_b}\left(\frac{\mu Tr^{\alpha_1}}{P_b}\right) = \exp\left(-2\pi\lambda \int_r^\infty \left(\frac{T}{T + (\frac{x}{r})^{\alpha_1}} \right) x \dd x \right),\]
and the result follows.
\end{proof}

The main difference between the two architectures is that in the three-node case the downlink user is not subject to any LI. Despite that, the downlink user is subject to intra-cell interference from the uplink users. Therefore, the SINR of $u_o$ at a random distance $r$ from $b_o$ in the three-node FD scenario is the same as \eqref{eq:sinr} but with $I_l = 0$.

\begin{theorem}\label{thm:three}
The outage probability of a downlink user in the three-node FD scenario is $\Pi_3(R, \lambda, \alpha_1, \alpha_2) = $
\begin{align}
1 - 2 \pi \lambda \int_0^\infty r e^{-\lambda \pi r^2-s\sigma_n^2}
\La_{I_b}\left(s\right) \La_{I_u}\left(s\right) \dd r,\label{eq:Thm2}
\end{align}
where $s = \frac{\mu Tr^{\alpha_1}}{P_b}$, $T=2^R-1$,
\begin{equation}
\La_{I_b}\left(s\right) = \exp\left(-2\pi\lambda \int_r^\infty \left(\frac{T}{T + (\frac{x}{r})^{\alpha_1}} \right) x \dd x \right),
\end{equation}
and
\begin{equation}
\La_{I_u}\left(s\right) = \exp\left(-2\pi\lambda \int_0^\infty \left(\frac{\frac{P_u}{P_b}T}{\frac{P_u}{P_b}T + \frac{y^{\alpha_2}}{r^{\alpha_1}}} \right) y \dd y \right).
\end{equation}
\end{theorem}

\begin{proof}
The proof is similar to the proof of Theorem \ref{thm:two} and so it is omitted due to space limitations.
\end{proof}

The derived expressions \eqref{eq:Thm1} and \eqref{eq:Thm2} provide a general result for the outage probability for each scenario under the main assumption that the interference is Rayleigh. Since these general expressions are not closed-form, a special case is considered which facilitates their simplification. Specifically, let $\alpha_1 = 4$ and $\alpha_2 = 4$. Furthermore, assume that the BSs and the users transmit with the same power, i.e. $P_b = P_u$, and that the network is interference-limited, i.e. $\sigma_n^2 = 0$. By using a series of transformations $w \rightarrow \frac{1}{\sqrt{T}}(\frac{x}{r})^2$, $z \rightarrow \frac{1}{\sqrt{T}}(\frac{y}{r})^2$, $u \rightarrow r^2$ and $v \rightarrow \rho^2$, the outage probability for the two-node scenario becomes,
\begin{equation}\label{eq:special}
\Pi_2(R, \lambda, 4, 4) = 1 - (\pi \lambda)^2 \int_0^\infty 
\frac{F(u,R)F(u,v,R)}{1+\sigma_l^2 Tu^2}\dd u,
\end{equation}
where
\begin{equation}
F(u,R) = \exp\left[-\pi\lambda u\left(1 + \sqrt{T} \arctan(\sqrt{T})\right) \right],
\end{equation}
and
\begin{align}\nonumber
&F(u,v,R)\\
&= \int_0^\infty \exp\left[-\pi \lambda\left(v + u \sqrt{T} \arccot\left(\frac{v}{u\sqrt{T}}\right)\right)\right]\dd v.
\end{align}\smallskip

Likewise, the outage probability for the three-node scenario can be simplified significantly to,
\begin{equation}
\Pi_3(R, \lambda, 4, 4) = 1 - \frac{1}{1+\sqrt{T}(\arctan(\sqrt{T})+\frac{\pi}{2})}.
\end{equation}
Note that $\Pi_3$ is independent from the network density $\lambda$ and only depends on the target rate $R$. The same applies for $\Pi_2$ when $\sigma_l^2=0$ even though it is not as obvious as for $\Pi_3$.

\section{Numerical Results}\label{sec:validation}
In this section, the proposed analytical model is validated and evaluated with computer simulations. Unless otherwise stated, the simulations use the following parameters: $\lambda = 10^{-3}$, $\alpha_1 = 4$, $\alpha_2 = 4$ and $P_b = P_u$. We compare the two FD models with the HD model in \cite{JA} which is similarly derived using stochastic geometry. For a fair comparison, we assume an RF-chain conserved framework \cite{AS6} where the HD and FD nodes use the same number of RF-chains and set the instantaneous achievable rate of the HD model to $\frac{1}{2}\log(1+\sinr)$ to accommodate the fact that the HD's instantaneous rate is half the one of the FD's due to the latter's simultaneous transmit/receive operation. Throughout this section, we will refer to the two-node scenario with no residual LI, i.e. $\sigma_l^2 = 0$, as optimal.

\begin{figure}[t]
  \centering
  \includegraphics[width=\linewidth]{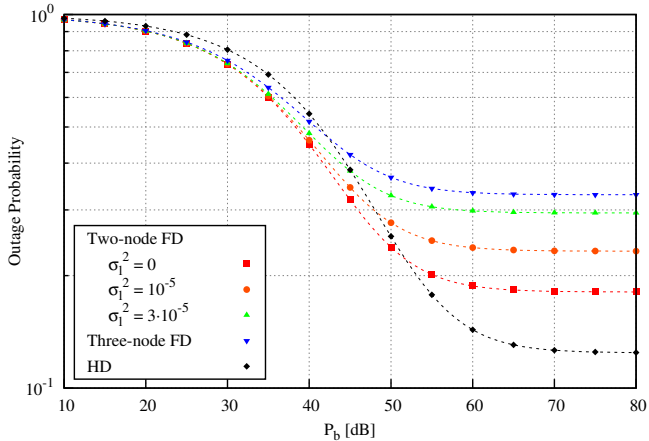}
  \caption{Outage Probability versus $P_b$; $\sigma_n^2=1$ and $R = 0.1$ bpcu. Analytical results are shown with dashed lines.}
  \label{fig:power}
\end{figure}

\begin{figure}[t]
  \includegraphics[width=\linewidth]{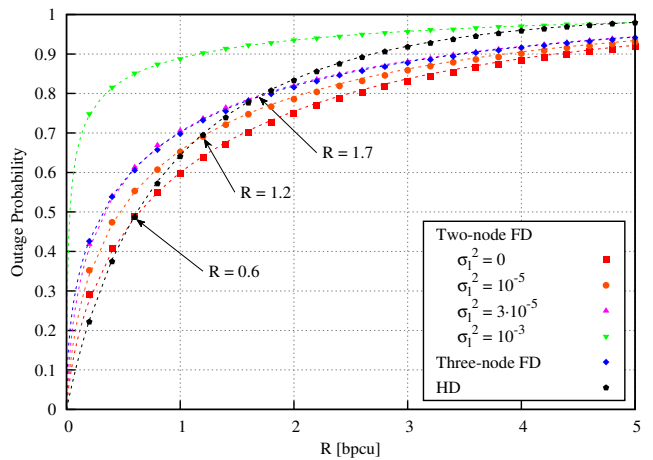}
  \caption{Outage Probability versus target rate $R$; $\sigma_n^2 = 0$ and $P_b = P_u$. Analytical results are shown with dashed lines.}
  \label{fig:bpcu}
\end{figure}

Fig. \ref{fig:power} shows the outage probability of each respective scenario in terms of the BSs' transmission power $P_b$. It is clear from the plot that the outage probability converges to a constant floor in all cases for high transmission powers. This is due to the fact that as the transmission power of the network's nodes increases, the noise in the network becomes negligible. The FD networks perform slightly better to the HD network since they can achieve twice the rate of the HD network. Nevertheless, the HD network suffers the least in terms of multiuser interference and therefore it performs significantly better for high transmission powers, whereas the performance of the FD networks is degraded by the interference. The three-node FD performs the worst for which shows the major impact the intra-cell interference has on the downlink. Indeed, the intra-cell interference starts to impact the performance at intermediate values, and as a result the outage converges faster than the other scenarios. On the other hand, the optimal two-node FD performs better than the three-node FD but its performance degrades as the residual LI increases. Since cellular networks are generally designed to be interference-limited, we will consider the case $\sigma_n^2 = 0$ for the rest of this section.

In Fig. \ref{fig:bpcu} the outage probability is depicted with respect to the target rate $R$. As expected, the performance for all scenarios degrades as the target rate increases. For low target rates, specifically up to $R = 0.6$ bits per channel use (bpcu), HD has the best performance. At $R = 0.6$ bpcu the performance of the HD network is equal to the one of the optimal two-node FD and at $R =1.7$ bpcu it is equal to the one of the three-node FD. This behavior is expected and is due to the fact that the FD can achieve twice the achievable rate at any instant but due to the multiuser interference and LI it achieves a better performance at higher target rates. Indeed, it is obvious that the residual LI has a critical impact on the performance since for the case $\sigma^2_l=10^{-3}$ the outage probability reaches 80\% for $R \approx 0.5$ bpcu. This is also clear from Fig. \ref{fig:loop}, which shows the negative impact of the residual LI on the performance of the network for different target rates. Moreover, as the target rate increases the performance drops with a faster pace.

\begin{figure}[t]
  \centering
  \includegraphics[width=\linewidth]{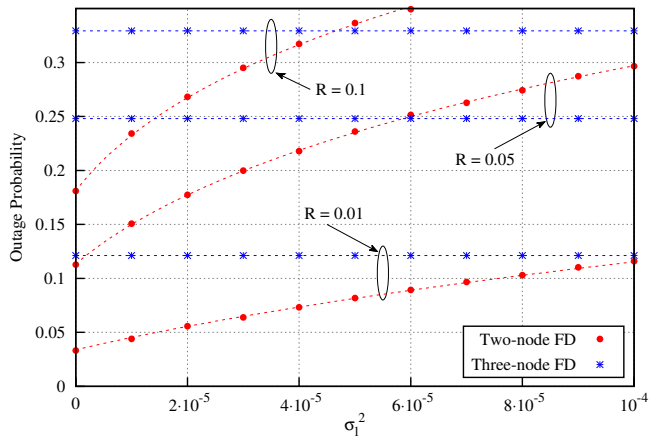}
  \caption{Outage Probability versus $\sigma_l^2$; $\sigma_n^2 = 0$ and $P_b = P_u$. Analytical results are shown with dashed lines.}
  \label{fig:loop}
\end{figure}

\begin{figure}[t]
  \centering
  \includegraphics[width=\linewidth]{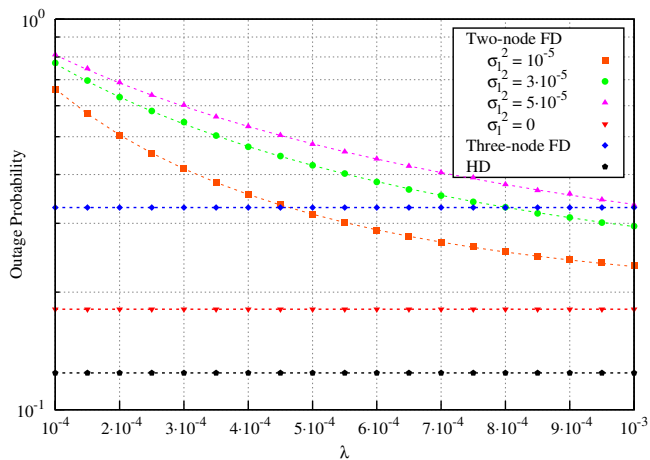}
  \caption{Outage Probability versus network density $\lambda$; $R = 0.1$ bpcu and $\sigma_n^2 = 0$. Analytical results are shown with dashed lines.}
  \label{fig:lambda}
\end{figure}
  
The impact of the network density $\lambda$ on the performance of the network is illustrated in Fig. \ref{fig:lambda}. The main observation here is that the outage probability for the optimal two-node FD, the three-node FD and the HD scenarios is independent of $\lambda$. We can explain this behavior as follows. Even though a larger (smaller) network density results to more (less) multiuser interference to the downlink user, it also entails that the user is closer (further) to its serving BS. This trade-off leads to the average performance to remain constant. However, when the user is subject to residual LI in the two-node FD scenario, the outage does depend on $\lambda$ and the denser the network, the lower the outage probability is. In this case, the residual LI dominates the SINR at the downlink user so the denser the network, the closer the user is to the BS thus reducing the negative effects of the residual LI. It is evident from Fig. \ref{fig:lambda} that the higher the power of the residual LI is, the denser the network needs to be in order to achieve the performance of the three-node FD scenario.

\section{Conclusion}\label{sec:conclusion}
This paper has presented analytical expressions for the outage probability of two fundamental FD architectures in cellular networks: two-node and three-node. A detailed performance comparison has been provided between these two architectures and their HD counterpart. Our results show that for low values of LI residual the two-node performs better than the three-node due to the latter's high multiuser interference. On the other hand, for large values of residual LI, the three-node becomes more practical. As expected, the HD mode performs better than the FD mode for low threshold values $R$. Therefore, even though both FD architectures have potential gains, to achieve these the multiuser interference and LI need to be reduced significantly. A future extension of this work is to consider the case where the nodes employ directional antennas which could reduce the multiuser interference and passively suppress the LI, thus improving the performance of both architectures.

\end{document}